\documentclass[pdflatex,sn-mathphys-num]{sn-jnl}

\usepackage{graphicx}%
\usepackage{multirow}%
\usepackage{amsmath,amssymb,amsfonts}%
\usepackage{amsthm}%
\usepackage{mathrsfs}%
\usepackage[title]{appendix}%
\usepackage{xcolor}%
\usepackage{textcomp}%
\usepackage{manyfoot}%
\usepackage{booktabs}%
\usepackage{algorithm}%
\usepackage{algorithmicx}%
\usepackage{algpseudocode}%
\usepackage{listings}%
\usepackage{enumerate}
\usepackage{epstopdf}
\usepackage{epsfig}
\usepackage{svg}
\usepackage{diagbox}
\usepackage{fancyhdr}

\addtocounter{MaxMatrixCols}{10}

\theoremstyle{thmstyleone}%
\newtheorem{theorem}{Theorem}
\theoremstyle{thmstyletwo}%
\newtheorem{example}{Example}%
\newtheorem{remark}{Remark}%

\theoremstyle{thmstylethree}%
\newtheorem{definition}{Definition}%

\theoremstyle{thmstyleone}%
\newtheorem{corollary}{Corollary}%

\begin{document}

\title[Article Title]{An Automated Theorem Generator with Theoretical Foundation Based on Rectangular Standard Contradiction}


\author[1,4]{\fnm{Yang} \sur{Xu}}\email{xuyang@swjtu.edu.cn}
\equalcont{All the authors are co-first authors.}

\author[2,4]{\fnm{Peiyao} \sur{Liu}}\email{liupeiyao@mail.xhu.edu.cn}
\equalcont{All the authors are co-first authors.}

\author[1,4]{\fnm{Shuwei} \sur{Chen}}\email{swchen@swjtu.edu.cn}
\equalcont{All the authors are co-first authors.}

\author[3,4]{\fnm{Jun} \sur{Liu}}\email{j.liu@ulster.ac.uk}
\equalcont{All the authors are co-first authors.}

\affil[1]{\orgdiv{School of Mathematics}, \orgname{Southwest Jiaotong University}, \orgaddress{\city{Chengdu}, \postcode{611756}, \state{Sichuan}, \country{China}}}

\affil[2]{\orgdiv{School of Computer and Software Engineering}, \orgname{Xihua University}, \orgaddress{\city{Chengdu}, \postcode{610039}, \state{Sichuan}, \country{China}}}

\affil[3]{\orgdiv{School of Computing}, \orgname{Ulster University}, \orgaddress{\city{Belfast}, \postcode{BT15 1ED}, \state{Northern Ireland}, \country{UK}}}

\affil[4]{\orgdiv{National-Local Joint Engineering Laboratory of System Credibility Automatic Verification}, \orgaddress{\city{Chengdu}, \postcode{611756}, \state{Sichuan}, \country{China}}}


\abstract{Currently, there is a lack of rigorous theoretical system for systematically generating non-trivial and logically valid theorems. Addressing this critical gap, this paper conducts research to propose a novel automated theorem generation theory and tool. Based on the concept of standard contradiction which possesses unique deductive advantages, this paper defines and proves, for the first time, a new logical structure known as rectangular standard contradiction. Centered on this structure, a complete Automated Theorem Generation (ATG) theory is put forward. Theoretical proofs clarify two core properties of rectangular standard contradiction: first, it is a standard contradiction (necessarily unsatisfiable); second, it exhibits non-redundancy (the remaining clause set becomes satisfiable after removing any clause). Leveraging these properties, this paper proves that partitioning a rectangular standard contradiction into a premise subset $A$ and negation of its complement $H$, a valid theorem $A \vdash \neg H$ can be formed, and all such theorems are logically equivalent. To implement this theory, an efficient template-based ATG algorithm is designed, and a Rectangular Automated Theorem Generator is developed. This research enables machines to transition from "verifiers" to "discoverers", opening up new avenues for fundamental research in the fields of logic and artificial intelligence.}




\maketitle

\section{Introduction}\label{sec1}

As a core branch of artificial intelligence and computer science, automated reasoning is dedicated to enabling computers to simulate human logical reasoning capabilities \cite{russell2021artificial}. However, for a long time, the research paradigm in the field of automated reasoning has mainly focused on the ``proving'' phase, i.e., verifying the validity of an existing theorem \cite{pantsar2024theorem, coward2022formal, stock2024application, hozzova2023program}. Regarding how to systematically and automatically "generate" or ``discover'' new, non-trivial, and logically necessary theorems starting from a set of basic literals, there is still a lack of a complete and rigorous theoretical system. This role transition from ``verifier'' to a ``discoverer'' is of crucial significance for enhancing the creativity of machine intelligence \cite{delgrande2024current, lin2024atg}.

To fill this theoretical gap, this paper proposes a complete theory of Automated Theorem Generation (ATG). The cornerstone of this theory is a novel logical structure - ``Rectangular Standard Contradiction'' - which we have proposed and rigorously defined for the first time. The theory presented in this paper is built upon the concept of Standard Contradiction\cite{xu2025contradictions}. A standard contradiction refers to an unsatisfiable clause set, whose structural characteristics offer unique advantages for logical deduction \cite{xu2018distinctive, xu2018contradiction}. In recent years, research on Standard Contradiction has achieved significant progress, leading to the formation series of systematic theories and algorithms \cite{xu2025extended, xu2025dynamic}. 

The core contribution of this paper lies in the establishment of a complete theory capable of directly generating correct theorems. Since the correctness of the theorem is guaranteed by their construction method and the inherent logical properties of Rectangular Standard Contradiction - with rigorous mathematical proofs provided in this paper - theorems generated through this method require no further verification by any theorem prover or manual effort.

First, this paper rigorously defines the ``Rectangular Standard Contradiction'' and conducts in-depth theoretical proofs regarding its properties. Theorem \ref{theorem1} proves that an $n$-level Rectangular Standard Contradiction constructed from $n$ generation literals is a standard contradiction, i.e., a necessarily unsatisfiable clause set. More importantly, Theorem \ref{theorem2} and its corollaries reveal its key ``non-redundant'' property: if any one or more clauses are removed from a complete Rectangular Standard Contradiction, the remaining clause set becomes satisfiable. This profound property directly forms the logical basis for theorem generation, as elaborated in Theorem \ref{theorem3}: taking any subset of the Rectangular Standard Contradiction as the premise ($A$), the negation of its complement ($\neg H$) constitutes a logically valid theorem ($A \vdash \neg H$). Theorem \ref{theorem4} further proves that all theorems generated by a set of generation literals are logically equivalent.

On the solid theoretical foundation, this paper further designs and implements an efficient Automated Theorem Generation (ATG) algorithm (Algorithm \ref{algo2}). By adopting an innovative ``template-based construction'' method (Algorithm \ref{algo1}), this algorithm avoids the complexity of naive construction method and reliance on the literal management. It can quickly construct the corresponding Rectangular Standard Contradiction from any given and qualified literal set, and automatically output a new theorem consisting of premises and a conclusion. Currently, we have released the first version of the automated theorem generator based on rectangular standard contradiction on GitHub (\url{https://github.com/lpy-2019/Automated-Theorem-Generator---Rectangle}), which readers can download and use.

The fundamental theoretical significance of this study lies in its that it provides the theoretical basis for computers to automatically ``generating'' or ``discovering'' paradigm. By providing a systematic method for constructing unsatisfiable sets and deriving theorems from them, this work opens up new perspectives for basic theoretical research in logic and artificial intelligence \cite{liu2025neural, ferrag2025llm}, and explores the possibility of machines performing logical creation \cite{abdelaziz2022learning}.

The structure of this paper is organized as follows: Section \ref{sec2} reviews preliminary knowledge such as first-order/propositional logic and the concept of standard contradiction. Section \ref{sec3} elaborates on the definition and properties of Rectangular Standard Contradiction, and presents a rigorous mathematical proof of its status as a standard contradiction. Section \ref{sec4} expounds the complete theory of theorem generation based on the properties of Rectangular Standard Contradiction. Section \ref{sec5} provides the implementation of an automated algorithm for theorem generation, from template construction to final theorem output. Finally, Section \ref{sec6} summarizes the entire paper and outlines future research directions.

\section{Preliminaries}\label{sec2}

This paper involves concepts from both first-order logic and propositional logic. This section briefly reviews the fundamental concepts of these two logical systems \cite{liu2023efficient, gleiss2020subsumption} and introduces the key notion of the standard contradiction, which forms the theoretical basis of our work.

\begin{definition}[Term, First-Order Logic]
	A \textbf{term} is a syntactic structure in first-order logic that denotes individual objects, defined recursively by the follows rules:
	
	\begin{itemize}
		\item Individual constants (e.g., $a, b, c$ and variables (e.g., $x, y, z$) are terms.
		\item If $f$ is an $n$-ary function symbol and $t_1, t_2, \cdots, t_n$ are terms, then $f(t_1, t_2, \cdots, t_n)$ is also a term.
	\end{itemize}
\end{definition}

Terms are used to construct atomic formulas, describing the properties or relationship of individual objects.

\begin{definition}[Literal]
	A \textbf{literal} is either an atomic formula (\textbf{Atom}) or its negation. In propositional logic, an atom is a propositional variable (e.g., $p, x_11$). In first-order logic, an atom is a predicate applied to terms (e.g., $P(x)$)
	
	\begin{itemize}
		\item \textbf{Positive literal}: The atomic formula itself (e.g., $P(x)$, $p$).
		\item \textbf{Negative literal}: The negation of an atomic formula (e.g., $\neg Q(a, b)$, $\neg p$, ).
	\end{itemize}
\end{definition}

Literals are the basic units for constructing clauses, used to express true of false assertions of propositions.

\begin{definition}[Clause]
	A \textbf{clause} is a disjunction (logical OR, $\vee$) of a finite set of literals.
	
	\begin{itemize}
		\item \textbf{Empty clause}: A clause containing no literals, denoted as $\square$.
	\end{itemize}
\end{definition}

Clauses are the core structure of the automated deduction rule.

\begin{definition}[CNF Formula]
	A \textbf{Conjunctive Normal Form} (in short, CNF) formula is a conjunction (logical AND, $\wedge$) of a finite set of clause, i.e., it has the form $C_1, C_2, \cdots, C_n$, where $C_i$ is a clause.
\end{definition}

Next is the definitions of the standard contradiction.

\begin{definition}[Standard Contradiction]\cite{xu2018contradiction}
	Suppose a clause set $\mathcal{S}=\{C_1, C_2, \cdots, C_m\}$ in propositional or first-order logic. If $\forall(l_{1},\cdots,l_{m}) \in \prod_{i=1}^{m} C_{i}$, there exits at least one complementary pair among ${l_1, l_2, \cdots, l_m}$, then $\mathcal{S} = \bigwedge_{i=1}^{m} C_{i}$ is called a \textbf{standard contradiction}.
\end{definition}

For a more detailed introduction to standard contradiction, please refer to reference \cite{xu2025contradictions}.

\section{The Rectangular Standard Contradiction}\label{sec3}

This section introduces and elaborates on the core of our theory - Rectangular Standard Contradiction. This novel logical construction is not only necessarily unsatisfiable, but more importantly, it possesses a unique ``non-redundant'' property, laying a solid logical foundation for the subsequent theorem generation theory.

\subsection{Concept}\label{subsec1}

We first start from a more general concept of ``maximal contradiction'' and then introduce ``rectangular standard contradiction'' as its special case.

\begin{definition}\label{def1}(Maximal Contradiction)
	Suppose a literal set $\mathcal{L}=\{l_1, l_2, \cdots, l_n\}$ in propositional or first-order logic,where $l_i$ is either an atom $p_i$ or a negated atom $\neg p_i$. A clause $C(l_1, l_2, \cdots, l_n) = \vee_{i=1}^{n}l_i$, a clause set $\mathcal{S} = \{C(l_1, l_2, \cdots, l_n) | l_i \in \{p_i, \neg p_i\}, i = 1, 2, \cdots, n\}$. Then $\mathcal{S}$ is a standard contradiction containing $n \times 2^n$ literals that includes $\mathcal{L}$, and $\mathcal{S}$ is called the \textbf{maximal contradiction} generated by $\mathcal{L}$. 
\end{definition}

Maximal contradiction provides us with a complete unsatisfiable set that covers all possibilities. However, its structure is overly broad. To obtain more refined properties, we introduce a highly structured and symmetric special case - rectangular standard contradiction.

According to Definition \ref{def1}, the length of each clause in the maximal contradiction $\mathcal{S}$ generated by $\mathcal{L}$ is $|L|$ (that is, the number of literals in $\mathcal{L}$), and there are a total of $2^n$ clauses in $\mathcal{S}$. If each of these $2^n$ clauses is arranged vertically is sequence, the form of $\mathcal{S}$ is a rectangle (with a length of $2^n$ literals and a height of $n$ literals). 

\begin{definition}\label{def2}[Rectangular Standard Contradiction]
	Suppose a literal set $\mathcal{L}=\{l_1, l_2, \cdots, l_n\}$ in propositional or first-order logic, and $\mathcal{R}_\mathcal{L}^n$ is the maximal contradiction generated by $\mathcal{L}$. When $\mathcal{R}_\mathcal{L}^n$ satisfies the following recursive rules, $\mathcal{R}_\mathcal{L}^n$ (here ``$n$'' means the number of literals in $\mathcal{L}$) is a \textbf{rectangular standard contradiction}.
	
	\begin{enumerate}[i)]
		\item Let $\mathcal{L}^\prime=\mathcal{L}\setminus\{l_n\}$, and $\mathcal{R}_{\mathcal{L}^\prime}^{n-1}$ is the rectangular standard contradiction generated by $\mathcal{L}^\prime$;
		\item The rectangular form of $\mathcal{R}_\mathcal{L}^n$ is as follows:
	\end{enumerate}
	
	\begin{equation}\label{equation1}
		\begin{array}{c}
			\begin{bmatrix}
				\hline
				\multicolumn{3}{|c|}{\mathcal{L}^\prime} & \multicolumn{3}{|c|}{\mathcal{L}^\prime} \\
				\hline
				l_n & \cdots & l_n & \neg l_n & \cdots & \neg l_n \\
			\end{bmatrix} \\
			\begin{array}{@{}c@{\hspace{1em}}c@{}}
				\underbrace{\hspace{4em}}_{2^{n-1}} & \underbrace{\hspace{4em}}_{2^{n-1}}
			\end{array}
		\end{array}		
	\end{equation}
	
	The literal set $\mathcal{L}$ is called the \textbf{generation literal set} of the rectangular standard contradiction $\mathcal{R}_\mathcal{L}^n$. The rectangular standard contradiction $\mathcal{R}_\mathcal{L}^n$ is also called \textbf{\emph{n}-level rectangular standard contradiction}.
\end{definition}

The recursive definition of Rectangular Standard Contradiction is the origin of all its excellent properties. It expands an ($n$-1)-level rectangular standard contradiction into an $n$-level structure in a systematic manner, ensuring a high degree of symmetry and analyzability of the whole.

\begin{remark}
	The generation literal set $\mathcal{L}$ must satisfy the rule that no identical predicate symbols exist. In particular, the equality symbol (``='') is a special type of predicate symbol.
\end{remark}

This rule is proposed to avoid the occurrence of complementary pairs in the clauses within a rectangular standard contradiction after substitution and to ensure that there are no redundant clauses in the rectangular standard contradiction $\mathcal{R}_\mathcal{L}^n$ generated by $\mathcal{L}$.

Each column of $\mathcal{R}_\mathcal{L}^n$ is a clause, and usually, the first column of $\mathcal{R}_\mathcal{L}^n$ is a clause formed by the generation literal set $\mathcal{L}$.

From Definition \ref{def2}, it can be known that $Length(\mathcal{L}^\prime) = 2^{n-1}$ (literals) and $Height(\mathcal{L}^\prime) = n - 1$ (literals). Combined Formula (\ref{equation1}), it can be directly derived that $Length(\mathcal{L}) = 2 \times Length(\mathcal{L}^\prime)$ and $Height(\mathcal{L}) = Height(\mathcal{L}^\prime) + 1$, where ``1'' means the last row $(l_n, \cdots, l_n, \neg l_n, \cdots, \neg l_n)$ of $\mathcal{R}_\mathcal{L}^n$.

Notably, in the last row $(l_n, \cdots, l_n, \neg l_n, \cdots, \neg l_n)$, the first $2^{n-1}$ literals are $l_n$, and the last $2^{n-1}$ literals are $\neg l_n$.

Next, we give two examples to illustrate the rectangular standard contradiction in propositional logic (Example \ref{example1}) and first-order logic (Example \ref{example2}).

\begin{example}\label{example1}
	Let $\mathcal{L} = \{w, x, y, z\}$ is a generation literal set in propositional logic. The 4-level rectangular standard contradiction $\mathcal{R}_\mathcal{L}^{4}$ generated by $\mathcal{L}$ is shown as follows. ($Height(\mathcal{R}_\mathcal{L}^{4}) = 4$, and $Length(\mathcal{R}_\mathcal{L}^{4}) = 2^4 = 16$)
	
	\begin{equation}\label{equation2}
		\begin{bmatrix}
		w & \neg w & w & \neg w & w & \neg w & w & \neg w & w & \neg w & w & \neg w & w & \neg w & w & \neg w \\
		x & x & \neg x & \neg x & x & x & \neg x & \neg x & x & x & \neg x & \neg x & x & x & \neg x & \neg x \\
		y & y & y & y & \neg y & \neg y & \neg y & \neg y & y & y & y & y & \neg y & \neg y & \neg y & \neg y \\
		z & z & z & z & z & z & z & z & \neg z & \neg z & \neg z & \neg z & \neg z & \neg z & \neg z & \neg z
		\end{bmatrix} 	
	\end{equation}
\end{example}

\begin{example}\label{example2}
	Let $\mathcal{L} = \{P_1(a), P_2(f(x)), P_3(g(y, a))\}$ is a generation literal set in first-order logic. The rectangular standard contradiction $\mathcal{R}_\mathcal{L}^{3}$ generated by $\mathcal{L}$ is shown as follows. ($Height(\mathcal{R}_\mathcal{L}^{3}) = 3$, and $Length(\mathcal{R}_\mathcal{L}^{3}) = 2^3 = 8$)
	
	\begin{equation}\label{equation3}
		\begin{bmatrix}
		P_1(a) & \neg P_1(a) & P_1(a) & \neg P_1(a) & P_1(a) & \neg P_1(a) & P_1(a) & \neg P_1(a) \\
		P_2(f(x)) & P_2(f(x)) & \neg P_2(f(x)) & \neg P_2(f(x)) & P_2(f(x)) & P_2(f(x)) & \neg P_2(f(x)) & \neg P_2(f(x)) \\
		P_3(g(y, a)) & P_3(g(y, a)) & P_3(g(y, a)) & P_3(g(y, a)) & \neg P_3(g(y, a)) & \neg P_3(g(y, a)) & \neg P_3(g(y, a)) & \neg P_3(g(y, a))
		\end{bmatrix} 	
	\end{equation}
\end{example}

\subsection{Naive Construction}\label{subsec2}

To gain a more intuitive understanding of the recursive structure in Definition \ref{def2}, we propose a step-by-step expansion construction method, namely the \textbf{naive construction method}. Starting from a single literal, this method iterates level by level until a complete $n$-level rectangular standard contradiction $\mathcal{R}_\mathcal{L}^{n}$ is constructed. This is not only a construction process, but more importantly, a procedural interpretation of the recursive definition. The steps of the naive method are follows.

Let $\mathcal{L} = \{x_1, x_2, \cdots, x_n\}$ is a generation literal set. For the convenience of illustration, the following sets are defined:

$\begin{array}{ll}
	\mathcal{L}_1 = \{x_1\} & \text{The element is the first literal in } \mathcal{L}. \\
	\mathcal{L}_2 = \{x_1, x_2\} & \text{The elements are the first 2 literals in } \mathcal{L}. \\
	\vdots \\
	\mathcal{L}_i = \{x_1, x_2, \cdots, x_i\} & \text{The elements are the first \emph{i} literals in } \mathcal{L}. \\
	\vdots \\
	\mathcal{L}_{n-1} = \{x_1, x_2, \cdots, x_{n-1}\} & \text{The elements are the first \emph{n} - 1 literals in } \mathcal{L}. \\
	\mathcal{L}_n = \{x_1, x_2, \cdots, x_n\} = \mathcal{L}. & \text{The elements are the first \emph{n} literals in } \mathcal{L} \text{, i.e, } \mathcal{L} \text{ itself}.
\end{array}$

\textbf{Step 1}: Using the $\mathcal{L}_1$ to construct a 1-level rectangular standard contradiction $\mathcal{R}_{\mathcal{L}_1}^1$ which contains $2^1$ clauses, then

\begin{equation}
	\mathcal{R}_{\mathcal{L}_1}^1 = 
	\begin{bmatrix}
		x_1 & \neg x_1
	\end{bmatrix} 	
\end{equation}

\textbf{Step 2}: Using the $\mathcal{L}_2$ and $\mathcal{R}_{\mathcal{L}_1}^1$ (from Step 1) to construct 2-level rectangular standard contradiction $\mathcal{R}_{\mathcal{L}_2}^2$ which contains $2^2$ clauses, then

\begin{equation}
	\mathcal{R}_{\mathcal{L}_2}^2 = 
	\begin{bmatrix}
		\hline
		\multicolumn{2}{|c|}{\mathcal{R}_{\mathcal{L}_1}^1} & \multicolumn{2}{|c|}{\mathcal{R}_{\mathcal{L}_1}^1} \\
		\hline
		x_2 & x_2 & \neg x_2 & \neg x_2
	\end{bmatrix}
	=
	\begin{bmatrix}
		x_1 & \neg x_1 & x_1 & \neg x_1 \\
		x_2 & x_2 & \neg x_2 & \neg x_2
	\end{bmatrix} 	
\end{equation}

\textbf{Step 3}: Using the $\mathcal{L}_3$  and $\mathcal{R}_{\mathcal{L}_2}^2$ (from Step 2) to construct 3-level rectangular standard contradiction $\mathcal{R}_{\mathcal{L}_3}^3$ which contains $2^3$ clauses, then

\begin{equation}
	\mathcal{R}_{\mathcal{L}_3}^3 = 
	\begin{bmatrix}
		\hline
		\multicolumn{4}{|c|}{\mathcal{R}_{\mathcal{L}_2}^2} & \multicolumn{4}{|c|}{\mathcal{R}_{\mathcal{L}_2}^2} \\
		\hline
		x_3 & x_3 & x_3 & x_3 & \neg x_3 & \neg x_3 & \neg x_3 & \neg x_3
	\end{bmatrix}
	=
	\begin{bmatrix}
		x_1 & \neg x_1 & x_1 & \neg x_1 & x_1 & \neg x_1 & x_1 & \neg x_1 \\
		x_2 & x_2 & \neg x_2 & \neg x_2 & x_2 & x_2 & \neg x_2 & \neg x_2 \\
		x_3 & x_3 & x_3 & x_3 & \neg x_3 & \neg x_3 & \neg x_3 & \neg x_3
	\end{bmatrix} 	
\end{equation}

$\cdots$

\textbf{Step \emph{i}}: Using the $\mathcal{L}_{i-1}$ and $\mathcal{R}_{\mathcal{L}_{i-1}}^{i-1}$ (from Step \emph{i}-1) to construct $i$-level rectangular standard contradiction $\mathcal{R}_{\mathcal{L}_i}^i$ which contains $2^{i-1}$ clauses, then

\begin{equation}
	\mathcal{R}_{\mathcal{L}_i}^i = 
	\begin{array}{c}
		\begin{bmatrix}
			\hline
			\multicolumn{4}{|c|}{\mathcal{R}_{\mathcal{L}_{i-1}}^{i-1}} & \multicolumn{4}{|c|}{\mathcal{R}_{\mathcal{L}_{i-1}}^{i-1}} \\
			\hline
			x_i & \cdots & x_i & \neg x_i & \cdots & \neg x_i \\
		\end{bmatrix} \\
		\begin{array}{@{}c@{\hspace{1em}}c@{}}
			\underbrace{\hspace{4em}}_{2^{i-1}} & \underbrace{\hspace{4em}}_{2^{i-1}}
		\end{array}
	\end{array}	
\end{equation}

$\cdots$

\textbf{Step \emph{n}}: Using the $\mathcal{L}$ and $\mathcal{R}_{\mathcal{L}_{n-1}}^{n-1}$ (from Step \emph{n}-1) to construct $n$-level rectangular standard contradiction $\mathcal{R}_{\mathcal{L}}^n$ which contains $2^n$ clauses, then

\begin{equation}
	\mathcal{R}_{\mathcal{L}}^n = 
	\begin{array}{c}
		\begin{bmatrix}
			\hline
			\multicolumn{4}{|c|}{\mathcal{R}_{\mathcal{L}_{n-1}}^{n-1}} & \multicolumn{4}{|c|}{\mathcal{R}_{\mathcal{L}_{n-1}}^{n-1}} \\
			\hline
			x_n & \cdots & x_n & \neg x_n & \cdots & \neg x_n \\
		\end{bmatrix} \\
		\begin{array}{@{}c@{\hspace{1em}}c@{}}
			\underbrace{\hspace{4em}}_{2^{n-1}} & \underbrace{\hspace{4em}}_{2^{n-1}}
		\end{array}
	\end{array}	
\end{equation}

\subsection{Proof}\label{subsec3}

To establish Rectangular Standard Contradiction as the logical cornerstone of our theory, it is first necessary to rigorously prove its core property: it is a standard contradiction, i.e., a necessarily unsatisfiable clause set. The following theorem completes this critical proof using mathematical induction.

\begin{theorem}\label{theorem1}
	Let $\mathcal{L}$ be a generation literal set containing \emph{n} literals, then the rectangular standard contradiction $\mathcal{R}_\mathcal{L}^n$ generated by $\mathcal{L}$ is a standard contradiction.
\end{theorem}

\begin{proof}
	We use mathematical induction on the number of literals in $\mathcal{L}$.
	
	\textbf{Base Case}: $n = 1$. 
	
	When $n = 1$, the generation literal set $\mathcal{L}_1$ contains exactly 1 literal, denoted $\mathcal{L}_1 = \{x\}$, then 
	
	\begin{center}
		$\mathcal{R}_{\mathcal{L}_1}^1$ = 
		$\begin{bmatrix}
			x & \neg x
		\end{bmatrix}$
		.
	\end{center}
	
	Obviously, $\mathcal{R}_{\mathcal{L}_1}^1$ is a standard contradiction with simplest structure. Thus, the theorem holds when $n = 1$.
	
	\textbf{Inductive Hypothesis}: Assume the theorem holds for any generation literal set $\mathcal{L}_k$ with $n = k$ literals (where $k \geq 1$). That is, for any $\mathcal{L}_k$ with $|\mathcal{L}_k| = k$, the rectangular standard contradiction $\mathcal{R}_{\mathcal{L}_k}^k$ generated by $\mathcal{L}_k$ is a standard contradiction, denoted as
	
	\begin{center}
		$\mathcal{R}_{\mathcal{L}_k}^k = \wedge_{i=1}^{2^k}C_i = $
		$\begin{bmatrix}
			x_{11} & x_{12} & \cdots & x_{1i} & \cdots & x_{12^k} \\
			x_{21} & x_{22} & \cdots & x_{2i} & \cdots & x_{22^k} \\
			\vdots & \vdots & \vdots & \vdots & \vdots & \vdots \\
			x_{k1} & x_{k2} & \cdots & x_{ki} & \cdots & x_{k2^k}
		\end{bmatrix}$
		.
	\end{center}
	
	There are $k$ rows and $2^k$ columns in $\mathcal{R}_{\mathcal{L}_k}^k$, and any column $C_i = \{x_{1i}, x_{2i}, \cdots, x_{ki}\}$ represents a clause. Thus, for any $(x^{[1]}, x^{[2]}, \cdots, x^{[2^k]}) \in \prod_{i=1}^{2^k}C_i$, there exists one complementary pair among $\{x^{[1]}, x^{[2]}, \cdots, x^{[2^k]}\}$. ($x^{[i]} \in C_i$, $i = 1, \cdots, 2^k$)
	
	\textbf{Inductive Step}: We now prove that the theorem holds for $n = k + 1$.
	
	Let $\mathcal{L}_{k+1}$ be a generation literal set with $n = k + 1$ literals. We can decompose $\mathcal{L}_{k + 1}$ as $\mathcal{L}_k \cup {x}$, where $\mathcal{L}_k$ is a subset of $\mathcal{L}_{k + 1}$ with $|\mathcal{L}_k| = k$ literals, and $x \notin \mathcal{L}_k$ (the ``new'' literal added to $\mathcal{L}_k$). Then the rectangular standard contradiction $\mathcal{R}_{\mathcal{L}_{k+1}}^{k+1}$ generated by $\mathcal{L}_{k+1}$ is denoted as follows:
	
	\begin{center}
		$\mathcal{R}_{\mathcal{L}_{k+1}}^{k+1}$ = 
		$\begin{array}{c}
			\begin{bmatrix}
				x_{11} & x_{12} & \cdots & x_{12^k} & x_{11} & x_{12} & \cdots & x_{12^k} \\
				x_{21} & x_{22} & \cdots & x_{22^k} & x_{21} & x_{22} & \cdots & x_{22^k} \\
				\vdots & \vdots & \vdots & \vdots & \vdots & \vdots & \vdots & \vdots \\
				x_{k1} & x_{k2} & \cdots & x_{k2^k} & x_{k1} & x_{k2} & \cdots & x_{k2^k} \\
				x & x & \cdots & x & \neg x & \neg x & \cdots & \neg x
			\end{bmatrix} \\
			\begin{array}{@{}c@{\hspace{1em}}c@{}}
				\underbrace{\hspace{7em}}_{2^{k}} & \underbrace{\hspace{7em}}_{2^{k}}
			\end{array}
		\end{array}$
		= 
		$\begin{array}{c}
			\begin{bmatrix}
				\hline
				\multicolumn{4}{|c|}{\mathcal{R}_{\mathcal{L}_k}^k} & \multicolumn{4}{|c|}{\mathcal{R}_{\mathcal{L}_k}^k} \\
				\hline
				x & x & \cdots & x & \neg x & \neg x & \cdots & \neg x
			\end{bmatrix} \\
			\begin{array}{@{}c@{\hspace{1em}}c@{}}
				\underbrace{\hspace{4em}}_{2^{k}} & \underbrace{\hspace{6em}}_{2^{k}}
			\end{array}
		\end{array}$
		.
	\end{center}
	
	for any $(x^{[1]}, x^{[2]}, \cdots, x^{[2^k]}, y^{[1]}, y^{[2]}, \cdots, y^{[2^k]}) \in \prod_{i=1}^{2^k}(C_i \cup \{x\}) \times \prod_{i=1}^{2^k}(C_i \cup \{\neg x\})$, four cases may arise, which we discuss one by one:
	
	a) If there no $x \in \{x^{[1]}, x^{[2]}, \cdots, x^{[2^k]}\}$, and exits $\neg x \in \{y^{[1]}, y^{[2]}, \cdots, y^{[2^k]}\}$, then there exists at least one complementary pair among $\{x^{[1]}, x^{[2]}, \cdots, x^{[2^k]}\}$. Thus, there exists one complementary pair among $(x^{[1]}, x^{[2]}, \cdots, x^{[2^k]}, y^{[1]}, y^{[2]}, \cdots, y^{[2^k]})$.
	
	b) If there no $\neg x \in \{y^{[1]}, y^{[2]}, \cdots, y^{[2^k]}\}$, and exists $x \in \{x^{[1]}, x^{[2]}, \cdots, x^{[2^k]}\}$, then there exists at least one complementary pair among $\{y^{[1]}, y^{[2]}, \cdots, y^{[2^k]}\}$. Thus, there exists one complementary pair among $(x^{[1]}, x^{[2]}, \cdots, x^{[2^k]}, y^{[1]}, y^{[2]}, \cdots, y^{[2^k]})$.
	
	c) If there no $x \in \{x^{[1]}, x^{[2]}, \cdots, x^{[2^k]}\}$, and exists $\neg x \in \{y^{[1]}, y^{[2]}, \cdots, y^{[2^k]}\}$, then there exists one complementary pair ($x, \neg x$) among $(x^{[1]}, x^{[2]}, \cdots, x^{[2^k]}, y^{[1]}, y^{[2]}, \cdots, y^{[2^k]})$. 
	
	d) If there no $x \in \{x^{[1]}, x^{[2]}, \cdots, x^{[2^k]}\}$, and exists $\neg x \in \{y^{[1]}, y^{[2]}, \cdots, y^{[2^k]}\}$, then there exists obviously one complementary pair among $(x^{[1]}, x^{[2]}, \cdots, x^{[2^k]}, y^{[1]}, y^{[2]}, \cdots, y^{[2^k]})$. 
	
	 In summary, for any $(x^{[1]}, x^{[2]}, \cdots, x^{[2^k]}, y^{[1]}, y^{[2]}, \cdots, y^{[2^k]}) \in \prod_{i=1}^{2^k}(C_i \cup \{x\}) \times \prod_{i=1}^{2^k}(C_i \cup \{\neg x\})$, there exists one complementary pair within it; that is, $\mathcal{R}_{\mathcal{L}_{k+1}}^{k+1}$ is a standard contradiction. Therefore, the theorem holds when $n = k + 1$.
	 
	 By the principle of mathematical induction, this theorem is thus proven.
\end{proof}

Each column of $\mathcal{R}_\mathcal{L}^n$ is a clause, $\mathcal{R}_\mathcal{L}^n$ is also a CNF formula. Furthermore, according to Theorem \ref{theorem1}, it follows that $\mathcal{R}_\mathcal{L}^n$ is UNSAT. The proof of Theorem \ref{theorem1} is of crucial importance, as it establishes the status of Rectangular Standard Contradiction as a logically necessarily false (UNSAT) clause set. Taking this as a starting point, we will explore its deeper structural properties in the next section, and these properties directly lead to the automated generation of theorems.

\section{Theorem Generation Based on Rectangular Standard Contradiction}\label{sec4}

In the previous section, we proved that Rectangular Standard Contradiction is a standard contradiction (UNSAT). This section will reveal its more profound "non-redundant" or "minimally unsatisfiable" property, and based on this, construct a complete theory of Automated Theorem Generation (ATG). This theory will clearly demonstrate how to systematically "extract" new, logically valid theorems from an unsatisfiable structure.

\begin{definition}\label{def3}
	A rectangular standard contradiction is referred to as a \textbf{full rectangular standard contradiction}, meaning that each column in the rectangle is filled with literals.
\end{definition}

Now, we introduce a key theorem that serves as a link between the preceding and subsequent content in our theoretical system, which reveals the non-redundancy of Rectangular Standard Contradiction.

\begin{theorem}\label{theorem2}
	There are no redundant clauses in a full rectangular standard contradiction (that is, the remaining clause set after removing any one clause from the full rectangular standard contradiction is satisfiable, i.e., \emph{SAT}).
\end{theorem}

\begin{proof}
	Let $\mathcal{R}_\mathcal{L}^n$ is a full rectangular standard contradiction where $\mathcal{L}$ is a generation literal set containing $n$ distinct literals.
	We still use mathematical induction on the number of literals in $\mathcal{L}$.
	
	\textbf{Base Case}: $n = 1$
	
	When $n = 1$, the generation literal set $\mathcal{L}_1$ contains exactly 1 literal, denoted as $\mathcal{L}_1 = \{x\}$, then
	
	\begin{center}
		$\mathcal{R}_{\mathcal{L}_1}^1$ = 
		$\begin{bmatrix}
			x & \neg x
		\end{bmatrix}$
		.
	\end{center}
	
	It is obviously that there are no redundant clauses in $\mathcal{R}_{\mathcal{L}_1}^1$, Plus, the theorem holds when $n = 1$.
	
	\textbf{Inductive Hypothesis}: Assume the theorem holds for any generation literal set $\mathcal{L}_k$ with $n = k$ literals (where $k \geq 1$). That is, for any $\mathcal{L}_k$ with $|\mathcal{L}_k| = k$, the rectangular standard contradiction $\mathcal{R}_{\mathcal{L}_k}^k$ generated by $\mathcal{L}_k$ is a standard contradiction, denoted as
	
	\begin{center}
		$\mathcal{R}_{\mathcal{L}_k}^k = \wedge_{i=1}^{2^k}C_i = $
		$\begin{bmatrix}
			x_{11} & x_{12} & \cdots & x_{1i} & \cdots & x_{12^k} \\
			x_{21} & x_{22} & \cdots & x_{2i} & \cdots & x_{22^k} \\
			\vdots & \vdots & \vdots & \vdots & \vdots & \vdots \\
			x_{k1} & x_{k2} & \cdots & x_{ki} & \cdots & x_{k2^k}
		\end{bmatrix}$
		.
	\end{center}
	
	There are $k$ rows and $2^k$ columns in $\mathcal{R}_{\mathcal{L}_k}^k$, and any column $C_i = \{x_{1i}, x_{2i}, \cdots, x_{ki}\}$ is a clause. After arbitrarily extracting one clause from $\mathcal{R}_{\mathcal{L}_k}^k$, the remaining clause set is satisfiable.
	
	\textbf{Inductive Step}: We now prove that the theorem holds for $n = k + 1$.
	
	Let $\mathcal{L}_{k+1}$ be a generation literal set with $n = k + 1$ literals. We can decompose $\mathcal{L}_{k + 1}$ as $\mathcal{L}_k \cup {x}$, where $\mathcal{L}_k$ is a subset of $\mathcal{L}_{k + 1}$ with $|\mathcal{L}_k| = k$ literals, and $x \notin \mathcal{L}_k$ (the ``new'' literal added to $\mathcal{L}_k$). Then the rectangular standard contradiction $\mathcal{R}_{\mathcal{L}_{k+1}}^{k+1}$ generated by $\mathcal{L}_{k+1}$ is denoted as follows:
	
	\begin{center}
		$\mathcal{R}_{\mathcal{L}_{k+1}}^{k+1}$ = 
		$\begin{array}{c}
			\begin{bmatrix}
				x_{11} & x_{12} & \cdots & x_{12^k} & x_{11} & x_{12} & \cdots & x_{12^k} \\
				x_{21} & x_{22} & \cdots & x_{22^k} & x_{21} & x_{22} & \cdots & x_{22^k} \\
				\vdots & \vdots & \vdots & \vdots & \vdots & \vdots & \vdots & \vdots \\
				x_{k1} & x_{k2} & \cdots & x_{k2^k} & x_{k1} & x_{k2} & \cdots & x_{k2^k} \\
				x & x & \cdots & x & \neg x & \neg x & \cdots & \neg x
			\end{bmatrix} \\
			\begin{array}{@{}c@{\hspace{1em}}c@{}}
				\underbrace{\hspace{7em}}_{2^{k}} & \underbrace{\hspace{7em}}_{2^{k}}
			\end{array}
		\end{array}$
		= 
		$\begin{array}{c}
			\begin{bmatrix}
				\hline
				\multicolumn{4}{|c|}{\mathcal{R}_{\mathcal{L}_k}^k} & \multicolumn{4}{|c|}{\mathcal{R}_{\mathcal{L}_k}^k} \\
				\hline
				x & x & \cdots & x & \neg x & \neg x & \cdots & \neg x
			\end{bmatrix} \\
			\begin{array}{@{}c@{\hspace{1em}}c@{}}
				\underbrace{\hspace{4em}}_{2^{k}} & \underbrace{\hspace{6em}}_{2^{k}}
			\end{array}
		\end{array}$
		.
	\end{center}
	
	Notably, $\mathcal{R}_{\mathcal{L}_{k+1}}^{k+1}$ is also a clause set, i.e., $\mathcal{R}_{\mathcal{L}_{k+1}}^{k+1} = \{\wedge_{i=1}^{2^k}(C_i \vee x)\} \wedge \{\wedge_{i=1}^{2^k}(C_i \vee \neg x)\}$.
	
	Arbitrarily extract a clause; for the sake of notational convenience, w.l.o.g, extract the first clause $C_1 \vee x$ from $\mathcal{R}_{\mathcal{L}_{k+1}}^{k+1}$ (since all clauses in the clause set $\mathcal{R}_{\mathcal{L}_{k+1}}^{k+1}$ are of equal status), denoted as follows:
	
	\begin{center}
		$\mathcal{R}_{\mathcal{L}_{k+1}}^{k+1} \setminus \{C_1 \vee x\}$ = 
			$\begin{array}{c}
			\begin{bmatrix}
				\colorbox{yellow!30}{$x_{12}$} & \colorbox{yellow!30}{$\cdots$} & \colorbox{yellow!30}{$x_{12^k}$} & x_{11} & x_{12} & \cdots & x_{12^k} \\
				\colorbox{yellow!30}{$x_{22}$} & \colorbox{yellow!30}{$\cdots$} & \colorbox{yellow!30}{$x_{22^k}$} & x_{21} & x_{22} & \cdots & x_{22^k} \\
				\colorbox{yellow!30}{$\vdots$} & \colorbox{yellow!30}{$\vdots$} & \colorbox{yellow!30}{$\vdots$} & \vdots & \vdots & \vdots & \vdots \\
				\colorbox{yellow!30}{$x_{k2}$} & \colorbox{yellow!30}{$\cdots$} & \colorbox{yellow!30}{$x_{k2^k}$} & x_{k1} & x_{k2} & \cdots & x_{k2^k} \\
				x & \cdots & x & \neg x & \neg x & \cdots & \neg x
			\end{bmatrix} \\
			\begin{array}{@{}c@{\hspace{1em}}c@{}}
				\underbrace{\hspace{6.5em}}_{2^{k}-1} & \underbrace{\hspace{7.5em}}_{2^{k}}
			\end{array}
		\end{array}$
		.
	\end{center}
	
	By the inductive hypothesis, the clause set $\wedge_{i=2}^{2^k}C_i$ composed of all literals marked in yellow is SAT (where $C_i = \{x_{1i}, x_{2i}, \cdots, x_{ki}\}$), meaning there exists a set of satisfying instances $Y = (y_2, \cdots, y_{2^k})$ that makes $\wedge_{i=1}^{2^k}C_i$ satisfied. It is not difficult to deduce that this satisfying instances $Y$ renders the clause set $\wedge_{i=2}^{2^k}(C_i \vee x)$ SAT. Furthermore, By adding $2^k$ instances of $\neg x$ to $Y$ results in a new literal set $Y^\prime = (y_2, \cdots, y_{2^k}, \neg x, \neg x, \cdots, \neg x)$, which renders the clause set $\mathcal{R}_{\mathcal{L}_{k+1}}^{k+1} \setminus \{C_1 \vee x\}$ SAT.
	
	It should also be noted that there are no complementary pairs in $Y = (y_2, \cdots, y_{2^k})$, so there are no complementary pairs in $Y^\prime = (y_2, \cdots, y_{2^k}, \neg x, \neg x, \cdots, \neg x)$ either. Moreover, the $2^{k+1} - 1 = 2^k - 1 + 2^k$ literals in $Y^\prime$ corresponding exactly to the literals in $C_2 \vee x, \cdots, C_{2^k} \vee x, C_1 \vee \neg x, \cdots, C_{2^k} \vee \neg x$, which means $Y\prime$ is a satisfying instance of the clause set $\mathcal{R}_{\mathcal{L}_{k+1}}^{k+1} \setminus \{C_1 \vee x\} = \{C_2 \vee x, \cdots, C_{2^k} \vee x, C_1 \vee \neg x, \cdots, C_{2^k} \vee \neg x\}$.
	
	Therefore, the theorem holds when $n = k + 1$.
	
	By the principle of mathematical induction, this theorem is thus proven.
\end{proof}

Theorem \ref{theorem2} reveals a profound dual property: a full rectangular standard contradiction as a whole is unsatisfiable, yet any proper subset of the clause set that constitutes the full rectangular standard contradiction is satisfiable. This property can be naturally extended to scenarios where multiple clauses are removed.

\begin{corollary}\label{corollary1}
	Let $\mathcal{R}_{\mathcal{L}^n}^n$ is a full rectangular standard contradiction. After removing any $k$ $(k \leq 2^n)$ clauses from $\mathcal{R}_{\mathcal{L}^n}^n$, the remaining clause set is satisfiable, i.e., \emph{SAT}.
\end{corollary}

\begin{proof}
	As can be seen from Theorem \ref{theorem2}, after any one clause from a full rectangular standard contradiction $\mathcal{R}_{\mathcal{L}^n}^n$, remaining clause set is SAT (that is, there will be no complementary pairs in the literal tuple formed by arbitrarily selecting $2^n - 1$ literals from the remaining $2^n - 1$ clauses).
	
	Then, after removing any $k$ clauses from $\mathcal{R}_{\mathcal{L}^n}^n$, there also is no complementary pairs in the literal tuple formed by arbitrarily selecting $2^n - k$ literals from the remaining $2^n - k$ clauses. Therefore, after removing any $k$ clauses from a rectangular standard contradiction, the remaining clause set is SAT.
	
	When $k = 2^n$, the remaining clause set is an empty set, and the empty set generally considered SAT.
\end{proof}

Theorem \ref{theorem1}, Theorem \ref{theorem2}, and Corollary \ref{corollary1} together form the logical foundation for theorem generation, and ensure that the theorems generated are meaningful - specifically, that such theorems are not derived from unsatisfiable clause sets.

\begin{theorem}\label{theorem3}
	For any full rectangular standard contradiction $\mathcal{R}_{\mathcal{L}}^n$, let $H = \{H_1, H_2, \cdots, H_k\}$ $(k \leq 2^n)$ be an arbitrary subset of $\mathcal{R}_{\mathcal{L}}^n$, and let $A = R \setminus H$. Then $A \vdash \neg H$ is theorem, where $A$ is the premise and $\neg H$ is the hypothesis.
\end{theorem}

\begin{remark}
	$A \vdash \neg H$ is referred to as a theorem generated by $\mathcal{R}_{\mathcal{L}}^n$.
\end{remark}

An intuitive interpretation of Theorem 3 is as follows: Since the premise $A$ combined with the conclusion $H$ forms a full rectangular standard contradiction ($A \wedge H  \equiv \text{UNSAT}$), this is equivalent to $A \wedge (\neg \neg H) \equiv \text{UNSAT}$, which can also be expressed as $A \wedge \neg (\neg H) \equiv \text{UNSAT}$. According to the deduction theorem, this indicates that starting from the premise $A$, the hypothesis $\neg H$ can necessarily be derived.

Furthermore, we find that all theorems generated from the same rectangular contradictory are logically equivalent, which greatly simplifies the selection and representation of theorems.

\begin{theorem}\label{theorem4}
	Let $\mathcal{L}$ is a generation literal set with $n$ literals, and $\mathcal{R}_{\mathcal{L}}^n$ be a full rectangular standard contradiction generated by $\mathcal{L}$. Then the theorems generated by $\mathcal{R}_{\mathcal{L}}^n$ are mutually equivalent.
\end{theorem}

\begin{proof}
	Let $C$ is an arbitrary clause in $\mathcal{R}_{\mathcal{L}}^n$, for $\forall r, k \leq 2^n$, let $\{E_1, E_2, \cdots, E_r\}$ and $D_1, D_2, \cdots, D_k$ be subsets of $\mathcal{R}_{\mathcal{L}}^n$. Then
	
	\begin{center}
		$\mathcal{R}_{\mathcal{L}}^n \setminus \{C\} \vdash \neg C$ iff $\mathcal{R}_{\mathcal{L}}^n$ is unsatisfiable iff $S \setminus \{E_1, E_2, \cdots, E_r\} \vdash \neg (E_1 \wedge E_2 \wedge \cdots \wedge E_r)$,
		
		$\mathcal{R}_{\mathcal{L}}^n \setminus \{C\} \vdash \neg C$ iff $\mathcal{R}_{\mathcal{L}}^n$ is unsatisfiable iff $S \setminus \{D_1, D_2, \cdots, D_k\} \vdash \neg (D_1 \wedge D_2 \wedge \cdots \wedge D_k)$,
		
		$\mathcal{R}_{\mathcal{L}}^n \setminus \{E_1, E_2, \cdots, E_r\} \vdash \neg \{E_1, E_2, \cdots, E_r\}$ iff $S \setminus \{D_1, D_2, \cdots, D_k\} \vdash \neg (D_1 \wedge D_2 \wedge \cdots \wedge D_k)$.
	\end{center}
	
	This theorem is thus proven.
\end{proof}

\begin{remark}\label{remark3}
	It follows from Theorem \ref{theorem4} that among all theorems generated by a full rectangular standard contradiction $\mathcal{R}_{\mathcal{L}}^n$, only one theorem needs to be selected for representation, i.e., $S \setminus \{C\} \vdash \neg C$.
\end{remark}

\section{Automated Theorem Generation Algorithm}\label{sec5}

The previous two sections have laid a solid theoretical foundation for Automated Theorem Generation (ATG). This section puts the theory into practice by designing and implementing an efficient automated algorithm. Taking any valid set of generating literals as input, this algorithm can automatically construct the corresponding Rectangular Standard Contradiction and output a new theorem consisting of premises and a conclusion. The core lies in an innovative ``template-based construction method'' we proposed, which is significantly superior to the naive construction method described in Section 3.2.

\subsection{Template-Based Construction Method}

A close observation of the structure of an $n$-level rectangular standard contradiction $\mathcal{R}_{\mathcal{P}_n}^n$ reveals that its inherent pattern is independent of specific literals and only related to the polarity of the literals (positive or negative). For instance, an $n$-level rectangular standard contradiction can be viewed as two ($n$-1)-level rectangular standard contradictions placed side by side, with a new row of literals and their negations appended below.

\begin{center}
	$\mathcal{R}_{\mathcal{L}_n}^n$ = 
	$\begin{array}{c}
		\begin{bmatrix}
			\hline
			\multicolumn{4}{|c|}{\mathcal{R}_{\mathcal{L}_{n-1}}^{n-1}} & \multicolumn{4}{|c|}{\mathcal{R}_{\mathcal{L}_{n-1}}^{n-1}} \\
			\hline
			x & x & \cdots & x & \neg x & \neg x & \cdots & \neg x
		\end{bmatrix} \\
		\begin{array}{@{}c@{\hspace{1em}}c@{}}
			\underbrace{\hspace{4em}}_{2^{n-1}} & \underbrace{\hspace{6em}}_{2^{n-1}}
		\end{array}
	\end{array}$
	.
\end{center}

We can abstract this polarity structure, using "!" to represent positive literals and "?" to represent negative literals, thereby obtaining a general 
n
-level Rectangular Standard Contradiction template. For any generating set containing 
n
literals, its Rectangular Standard Contradiction follows the same 
n
-level template.

The literals directly below $\mathcal{R}_{\mathcal{L}_{n-1}}^{n-1}$ on the left are $2^{n-1}$ generation literals $x$, and the literals directly below $\mathcal{R}_{\mathcal{L}_{n-1}}^{n-1}$ on the right are the opposite counterparts of these $2^{n-1}$ generation literals $x$, i.e., $\neg x$. We cab abstract this polarity structure, using ``!'' to represent positive literals and ``?'' to represent negative literals, thereby obtaining a general $n$-level rectangular standard contradiction template. For any set of generation literals containing $n$ literals, its rectangular standard contradiction follows the same $n$-level template. The structure of $\mathcal{R}_{\mathcal{L}_n}^n$ can be transformed as follows:

\begin{equation}\label{formula9}
		\begin{array}{c}
		\begin{bmatrix}
			! & ? & \cdots & ? & ! & ? & \cdots & ? \\
			! & ! & \cdots & ? & ! & ! & \cdots & ? \\
			\vdots & \vdots & \vdots & \vdots & \vdots & \vdots & \vdots & \vdots \\
			! & ! & \cdots & ? & ! & ! & \cdots & ? \\
			! & ! & \cdots & ! & ? & ? & \cdots & ?
		\end{bmatrix} \\
		\begin{array}{@{}c@{\hspace{1em}}c@{}}
			\underbrace{\hspace{3em}}_{2^{n-1}} & \underbrace{\hspace{3em}}_{2^{n-1}}
		\end{array}
	\end{array}
\end{equation}

The structure in Formula (\ref{formula9}) is called an \textbf{$\textbf{\emph{n}}$-level rectangular standard contradiction template}. That is, the template of the rectangular standard contradiction generated by generation literal set containing $n$ literals is the same $n$-level template.

\begin{example}
	The 3-level and 4-level rectangular standard contradiction templates are as follows.
	
	\vspace{0.5\baselineskip}
	
	3-level template: 
	$\begin{bmatrix}
		! & ? & ! & ? & ! & ? & ! & ?\\
		! & ! & ? & ? & ! & ! & ? & ?\\
		! & ! & ! & ! & ? & ? & ? & ?
	\end{bmatrix}$;
	
	\vspace{0.5\baselineskip}
	
	4-level template:
	$\begin{bmatrix}
		! & ? & ! & ? & ! & ? & ! & ? & ! & ? & ! & ? & ! & ? & ! & ?\\
		! & ! & ? & ? & ! & ! & ? & ? & ! & ! & ? & ? & ! & ! & ? & ?\\
		! & ! & ! & ! & ? & ? & ? & ? & ! & ! & ! & ! & ? & ? & ? & ?\\
		! & ! & ! & ! & ! & ! & ! & ! & ? & ? & ? & ? & ? & ? & ? & ?
	\end{bmatrix}$.
\end{example}

Compared with the naive construction method, the template-based method completely separates structure generation from content population. We first efficiently generate an abstract polarity template, and then "populate" the specific generating literal set into the template. This method avoids the complex management of the literals themselves during the construction process, making the algorithm logic clearer and the execution efficiency higher.

Algorithm \ref{algo1} is the pseudo-code for generating an $n$-level template. By means of an iterative approach, the template is constructed incrementally from the 1-level template up to the $n$-level template. Each step only process the current level, resulting in extremely high efficiency and clear logic.

\begin{algorithm}
	\caption{$n$-level Template Construction Method}\label{algo1}
	\begin{algorithmic}[1]
		\Require 
		A positive integer $n$ indicating the template level. 
		\Ensure An $n$-level full rectangular standard contradiction template (2D array).
		\If{$n \le 0$} 
		\State \Return \emph{empty list}
		\EndIf
		
		\Comment{1. Initialization: Start with the 1-level template}
		\State $\emph{current\_template} \gets [\ ``!'', ``?'']\ ]$
		
		\Comment{2. Iterative construction: Build up from 2-level to $n$-level}
		\For{$k \gets 2$ \textbf{to} $n$}
		\Comment{a. Duplicate and extend: Copy and concatenate each row of the previous level template}
		\State $\emph{extended\_rows} \gets$ \emph{empty list}
		\For{\textbf{each} \emph{row} $\in \emph{current\_template}$}
		\State $\emph{extended\_rows}\text{.append}(\emph{row} \parallel \emph{row})$
		\EndFor
		
		\Comment{b. Generate new row: Create a new row with $2^{k-1}$ ``!'' followed by $2^{k-1}$ ``?``}
		\State $\emph{half\_length} \gets 2^{k-1}$
		\State $\emph{new\_row} \gets (["!"] \times \emph{half\_length}) \parallel (["?"] \times \emph{half\_length})$
		
		\Comment{c. Merge: Combine the extended rows and the new row to form the current $k$-level template}
		\State $\emph{current\_template} \gets \emph{extended\_rows}$
		\State $\emph{current\_template}.\textsc{Append}(\emph{new\_row})$
		\EndFor
		
		\Comment{3. Return the final result}
		\State \Return $\emph{current\_template}$
	\end{algorithmic}
\end{algorithm}

\subsection{ATG Algorithm}

Combined with Algorithm \ref{algo1}, we finally propose a complete Automated Theorem Generation (ATG) algorithm (i.e., Algorithm \ref{algo2}). This algorithm encapsulates the entire process from theory to practice: it receives a generating literal set, constructs a Rectangular Standard Contradiction via the template, automatically divides the premises and conclusion according to Theorem \ref{theorem3} and Remark \ref{remark3}, and finally outputs a brand-new theorem. Eventually, based on the ATG algorithm, we have implemented the \textbf{Automated Theorem Generator - Rectangle}, whose first version has been released on GitHub (\url{https://github.com/lpy-2019/Automated-Theorem-Generator---Rectangle}). This tool features a clear GUI, and anyone with basis mathematical logic knowledge can quickly get started with in.

\begin{algorithm}
	\caption{Automated Theorem Generation (ATG)}\label{algo2}
	\begin{algorithmic}[1]
		\Require 
		A literal set $\mathcal{L}$.
		\Ensure 
		A theorem (clause set).
		
		\State $n \gets |\mathcal{L}|$ 
		\Comment{Get the number of literals in $\mathcal{L}$}
		\State \emph{num\_clauses} $\gets 2^n$  
		\Comment{Get the number of clauses in the rectangular standard contradiction}
		\State $\mathcal{R} \gets $ [ ]
		\Comment{Initialization of the rectangular standard contradiction (clause set form)}
		
		\Comment{Call Algorithm 1 to generate a $n$-level template (2D array form)}
		\State \emph{template} $\gets$ \textsc{Make\_template}($n$)
		
		\For{$i \gets 1$ \textbf{to} \emph{num\_clause}}
		\Comment{Constructed clauses column by column}
		\State \emph{col\_array} $\gets$ \emph{template}$[i]$
		\State \emph{clause} $\gets$ [ ]
		\Comment{Initialization of a clause}
		\For{$j \gets 1$ \textbf{to} $n$}
		\Comment{Fill clauses row by by row}
		\If{\emph{col\_array}[$j$] $\neq$ ``?''}
		\Comment{Fill in the corresponding positive literal}
		\State \emph{clause}.\textsc{Append}(\textsc{Generating\_literal}($\mathcal{R}[j]$))
		\Else
		\Comment{Fill in the corresponding negative literal}
		\State \emph{clause}.\textsc{Append}(\textsc{Negation}(\textsc{Generating\_literal}($\mathcal{R}[j]$)))
		\EndIf
		\EndFor
		\State $\mathcal{R}$.\textsc{Append}(\emph{clause})
		\EndFor
		
		\State \emph{hypothesis} $\gets$ \textsc{Negation}($\mathcal{R}[1]$)
		\Comment{Negate $\mathcal{R}$'s 1st clause as the theorem's hypothesis}
		\State \emph{premises} $\gets \{\mathcal{R}[2], \mathcal{R}[3], \cdots, \mathcal{R}[2^n]\}$
		\Comment{Remaining clauses are the theorem's premises}
		
		\State \emph{theorem} $\gets \{hypothesis, premises\}$
		\Comment{Get the theorem}
		
		\State \Return $\emph{theorem}$
	\end{algorithmic}
\end{algorithm}

\section{Conclusions and Future Work}\label{sec6}

This paper introduces a complete and rigorous theory for Automated Theorem Generation (ATG), and focus on theorem discovery by formalizing the ``Rectangular Standard Contradiction'', a novel, recursively defined logical structure. We have proven that this structure is a non-redundant unsatisfiable clause set, which provides a systematic method for generating new, valid theorems by partitioning the structure into premises and a conclusion.

The significance of this work lies in providing a complete theoretical framework, complemented by an efficient template-based algorithm, that empowers machines to transition from logical `verifiers' to `discoverers'. This paradigm shift opens new possibilities for machine creativity and advances fundamental research in logic and artificial intelligence, establishing a solid foundation for the automated discovery of new knowledge.

As our conclusions remain valid when logical literals are replaced with first-order logic formulas, our primary focus for future work will be extending the automated theorem generator to accept first-order closed formulas as input, thereby significantly broadening its application scenarios. Additionally, we aim to deeply integrate domain expertise - particularly from mathematics - with our automated theorem generator to establish a technical pathway capable of generating framework with practical semantic value. Finally, there remains considerable scope for meaningful theoretical research in automated theorem generation, which constitutes another important direction for our future endeavors.

\section*{Declarations}

\noindent
\textbf{Funding} This work has been supported by the Key Project of Sichuan Science and Technology Innovation and Entrepreneurship Seeding Program (Grant No. 2024JDRC0084).

\textbf(Competing interests) The authors declare no competing interests.


\bibliography{sn-bibliography}

\end{document}